\documentclass[12pt,abstract]{scrartcl}
\usepackage[dvipdfmx]{graphicx}
\usepackage{amsmath,amsthm}
\makeatletter
\renewcommand{\theenumii}{\@roman\c@enumii}
\makeatother

\usepackage{natbib}
\newtheorem{lemma}{Lemma}
\newtheorem{theorem}{Theorem}

\usepackage{amssymb}
\usepackage{bm}
\usepackage{mathptmx}
\begin{document}

\title{On the relation between Sion's minimax theorem and existence of Nash equilibrium in asymmetric multi-players zero-sum game with only one alien\thanks{This work was supported by Japan Society for the Promotion of Science KAKENHI Grant Number 15K03481 and 18K01594.}}

\author{%
Atsuhiro Satoh\thanks{atsatoh@hgu.jp}\\[.01cm]
Faculty of Economics, Hokkai-Gakuen University,\\[.02cm]
Toyohira-ku, Sapporo, Hokkaido, 062-8605, Japan,\\[.01cm]
\textrm{and} \\[.1cm]
Yasuhito Tanaka\thanks{yasuhito@mail.doshisha.ac.jp}\\[.01cm]
Faculty of Economics, Doshisha University,\\
Kamigyo-ku, Kyoto, 602-8580, Japan.\\}

\date{}

\maketitle
\thispagestyle{empty}

\vspace{-1.5cm}

\begin{abstract}
We consider the relation between Sion's minimax theorem for a continuous function and a Nash equilibrium in an asymmetric multi-players zero-sum game in which only one player is different from other players, and the game is symmetric for the other players. Then,
\begin{enumerate}
	\item The existence of a Nash equilibrium,  which is symmetric for players other than one player, implies Sion's minimax theorem for pairs of this player and one of other players with symmetry for the other players.
	\item Sion's minimax theorem for pairs of one player and one of other players with symmetry for the other players implies the existence of a Nash equilibrium which is symmetric for the other players.
\end{enumerate}
Thus, they are equivalent. 

\end{abstract}

\begin{description}
	\item[Keywords:] multi-players zero-sum game, one alien, Nash equilibrium, Sion's minimax theorem
\end{description}

\begin{description}
	\item[JEL Classification:] C72
\end{description}

\section{Introduction}

We consider the relation between Sion's minimax theorem for a continuous function and a Nash equilibrium in an asymmetric multi-players zero-sum game in which only one player is different from other players, and the game is symmetric for the other players. We will show the following results.
\begin{enumerate}
	\item The existence of a Nash equilibrium,  which is symmetric for players other than one player, implies Sion's minimax theorem for pairs of this player and one of other players with symmetry for the other players.
	\item Sion's minimax theorem for pairs of one player and one of other players with symmetry for the other players implies the existence of a Nash equilibrium which is symmetric for the other players.
\end{enumerate}
Thus, they are equivalent. Symmetry for the other players means that those players (players other than one player) have the same payoff function and strategy space, and so their equilibrium strategies, maximin strategies and minimax strategies are the same.

An example of such a game is a relative profit maximization game in a Cournot oligopoly. Suppose that there are four firms, A, B, C and D in an oligopolistic industry. Let $\bar{\pi}_A$, $\bar{\pi}_B$, $\bar{\pi}_C$ and  $\bar{\pi}_D$ be the absolute profits of the firms. Then, their relative profits are
\begin{align*}
&\pi_A=\bar{\pi}_A-\frac{1}{3}(\bar{\pi}_B+\bar{\pi}_C+\bar{\pi}_D),\ \pi_B=\bar{\pi}_B-\frac{1}{3}(\bar{\pi}_A+\bar{\pi}_C+\bar{\pi}_D),\\
&\pi_C=\bar{\pi}_C-\frac{1}{3}(\bar{\pi}_A+\bar{\pi}_B+\bar{\pi}_D),\ \pi_D=\bar{\pi}_D-\frac{1}{3}(\bar{\pi}_A+\bar{\pi}_B+\bar{\pi}_C).
\end{align*}
We see
\[\pi_A+\pi_B+\pi_C+\pi_D=\bar{\pi}_A+\bar{\pi}_B+\bar{\pi}_C+\pi_D-(\bar{\pi}_A+\bar{\pi}_B+\bar{\pi}_C+\pi_C)=0.\]
Thus, the relative profit maximization game in a Cournot oligopoly is a zero-sum game\footnote{About relative profit maximization under imperfect competition please see \cite{mm}, \cite{ebl2}, \cite{eb2}, \cite{st}, \cite{eb1}, \cite{ebl1} and \cite{redondo}}. If the oligopoly is fully asymmetric because the demand function is not symmetric (in a case of differentiated goods) or firms have different cost functions (in both homogeneous and differentiated goods cases), maximin strategies and minimax strategies of firms do not correspond to Nash equilibrium strategies. However, if the oligopoly is symmetric for three firms in the sense that the demand function is symmetric and those firms have the same cost function, the maximin strategies of those firms with the corresponding minimax strategy of one firm (for the other players) constitute a Nash equilibrium which is symmetric for the three firms. In Appendix we present an example of a four-firms relative profit maximizing oligopoly. We see from this example that with two aliens the equivalence result does not hold.

\section{The model and Sion's minimax theorem}

Consider a multi-players zero-sum game with only one alien. There are $n$ players $i=1, \dots, n$, $n\geq 3$.  The strategic variables for the players are $s_1$, $s_2$, \dots, $s_n$, and ($s_1, s_2, \dots, s_n)\in S_1\times S_2\times \dots \times S_n$. $S_1$, $S_2,\ \dots,\ S_n$ are convex and compact sets in linear topological spaces. The payoff function of each player is $u_i(s_1, s_2,\dots, s_n), i=1, 2, \dots, n$. We assume 
\begin{quote}
$u_i$'s for $i=1, 2, \dots, n$ are continuous real-valued functions on $S_1\times S_2\times \dots \times S_n$, quasi-concave on $S_i$ for each $s_j\in S_j,\ j\neq i$, and quasi-convex on $S_j$ for $j\neq i$ for each $s_i\in S_i$.
\end{quote}
$n$ players are partitioned into two groups. Group 1 and Group $n$. Group 1 includes $n-1$ players, Players 1, 2, $\dots$, $n-1$, and Group $n$ includes only Player $n$. In Group 1 $n-1$ players are symmetric in the sense that  they have the same payoff function and strategy space. Thus, their equilibrium strategies, maximin strategies and minimax strategies are the same. Only Player $n$ has a different payoff function and a strategy space. Its equilibrium strategy may be different from those for the other players.

Since the game is a zero-sum game, we have
\begin{equation}
u_1(s_1, s_2,\dots,  s_n)+u_2(s_1, s_2,\dots,  s_n)+\dots, u_n(s_1, s_2,\dots,  s_n)=0,\label{e1}
\end{equation}
for given $(s_1, s_2,\dots, s_n)$.

Sion's minimax theorem (\cite{sion}, \cite{komiya}, \cite{kind}) for a continuous function is stated as follows.
\begin{lemma}
Let $X$ and $Y$ be non-void convex and compact subsets of two linear topological spaces, and let $f:X\times Y \rightarrow \mathbb{R}$ be a function, that is continuous and quasi-concave in the first variable and continuous and quasi-convex in the second variable. Then
\[\max_{x\in X}\min_{y\in Y}f(x,y)=\min_{y\in Y}\max_{x\in X}f(x,y).\] \label{l1}
\end{lemma}
We follow the description of this theorem in \cite{kind}.

Let $s_j$'s for $j\neq i, n;\ i,j \in \{1, 2, \dots, n-1\}$ be given. Then, $u_i(s_1, s_2, \dots, s_n)$ is a function of $s_i$ and $s_n$. We can apply Lemma \ref{l1} to such a situation, and get the following equation.
\begin{equation}
\max_{s_i\in S_i}\min_{s_n\in S_n}u_i(s_1, s_2, \dots, s_n)=\min_{s_n\in S_n}\max_{s_i\in S_i}u_i(s_1, s_2, \dots, s_n).\label{as0}
\end{equation}
Note that we do not require
\[\max_{s_n\in S_n}\min_{s_i\in S_i}u_n(s_1, s_2, \dots, s_n)=\min_{s_i\in S_i}\max_{s_n\in S_n}u_n(s_1, s_2, \dots, s_n),\]
nor
\[\max_{s_i\in S_i}\min_{s_j\in S_j}u_i(s_1, s_2, \dots, s_n)=\min_{s_j\in S_j}\max_{s_i\in S_i}u_i(s_1, s_2, \dots, s_n),\ j\neq i;\ i,j\in \{1, 2, \dots, n-1\}.\]

We assume that $\arg\max_{s_i\in S_i}\min_{s_n\in S_n}u_i(s_1, s_2, \dots, s_n)$ and $\arg\min_{s_n\in S_n}\max_{s_i\in S_i}u_i(s_1, s_2, \dots, s_n)$ are unique, that is, single-valued. By the maximum theorem they are continuous in $s_j$'s,$\ j\neq i, n$. {Also, throughout this paper we assume that the maximin strategy and the minimax strategy of players in any situation are unique, and the best responses of players in any situation are unique.}

Let $s_j=s$ for all $j\neq i,\ j\in \{1,2, \dots, n-1\}$. Consider the following function.
\[s\rightarrow \arg\max_{s_i\in S_i}\min_{s_n\in S_n}u_i(s, \dots, s_i, \dots, s, \dots, s_n).\]
Since $u_i$ is continuous, $S_i$ and $S_n$ are compact, and all $S_i$'s are the same, this function is also continuous with respect to $s$. Thus, there exists a fixed point. Denote it by $\tilde{s}$. $\tilde{s}$ satisfies
\begin{equation}
\arg\max_{s_i\in S_i}\min_{s_n\in S_n}u_i(\tilde{s}, \dots, s_i, \dots, \tilde{s}, \dots, s_n)=\tilde{s}.\label{fix1}
\end{equation}
From (\ref{as0}) we have
\begin{equation}
\max_{s_i\in S_i}\min_{s_n\in S_n}u_i(\tilde{s}, \dots, s_i, \dots, \tilde{s}, \dots, s_n)=\min_{s_n\in S_n}\max_{s_i\in S_i}u_i(\tilde{s}, \dots, s_i, \dots, \tilde{s}, \dots, s_n).\label{fix2}
\end{equation}
From symmetry for Players 1, 2, $\dots$, $n-1$, $\tilde{s}$ satisfies (\ref{fix1}) and (\ref{fix2}) for all $i\in \{1, 2, \dots, n-1\}$.

\section{The main results}

Consider a Nash equilibrium of an $n$-players zero-sum game. Let $s_i^*$'s, $i\in \{1, 2, \dots, n-1\}$ and $s^*_n$, be the values of $s_i$'s which, respectively, maximize $u_i$'s. Then,
\begin{equation*}
u_i(s_1^*,\dots, s_i^*, \dots, s_n^*)\geq u_i(s_1^*, \dots, s_i, \dots, s_n^*)\ \mathrm{for\ all}\ s_i\in S_i,\ i\in \{1, 2, \dots, n-1\},
\end{equation*}
\begin{equation*}
u_n(s_1^*, s_2^*, \dots, s_{n-1}^*, s_n^*)\geq u_n(s_1^*, s_2^*, \dots, s_{n-1}^*,s_n)\ \mathrm{for\ all}\ s_n\in S_n.
\end{equation*}
They mean
\[\arg\max_{s_i\in S_i}u_i(s_1^*, \dots, s_i, \dots, s_n^*)=s_i^*,\ i\in \{1, 2, \dots, n-1\},\]
and
\[\arg\max_{s_n\in S_n}u_n(s_1^*, \dots, s_i^*, \dots, s_n)=s_n^*.\]

We assume that the Nash equilibrium is symmetric in Group 1 that is, it is symmetric for Players 1, 2, \dots, $n-1$. Then, $s_i^*$'s are the same and $u_i(s_1^*, \dots, s_i^*, \dots, s_n^*)$'s are equal for all $i\in \{1, 2, \dots, n-1\}$. Also we have
\[u_i(s^*_1, \dots, s_i^*, \dots, s_j^*, \dots, s_n)=u_j(s^*_1, \dots, s_i^*, \dots, s_j^*, \dots, s_n),\ j\neq i;\ i, j\in \{1, 2, \dots, n-1\}.\]
Since the game is zero-sum,
\[\sum_{i=1}^{n-1}u_i(s^*_1, \dots, s_i^*, \dots, s_j^*, \dots, s_n)=(n-1)u_i(s^*_1, \dots, s_i^*, \dots, s_j^*, \dots, s_n)=-u_n(s^*_1, \dots, s_i^*, \dots, s_j^*, \dots, s_n).\]
Thus,
\[\arg\min_{s_n\in S_n}u_i(s^*_1, \dots, s_i^*, \dots, s_j^*, \dots, s_n)=\arg\max_{s_n\in S_n}u_n(s^*_1, \dots, s_i^*, \dots, s_j^*, \dots, s_n)=s_n^*,\]
This implies
\begin{align*}
&\min_{s_n\in S_n}u_i(s^*_1, \dots, s_i^*, \dots, s_j^*, \dots, s_n)=u_i(s^*_1, \dots, s_i^*, \dots, s_j^*, \dots, s_n^*)\\
&=\max_{s_i\in S_i}u_i(s^*_1, \dots, s_i, \dots, s_j^*, \dots, s_n^*).
\end{align*}

First we show the following theorem.
\begin{theorem}
The existence of a Nash equilibrium, which is symmetric in Group 1, implies Sion's minimax theorem for pairs of a player in Group 1 and Player $n$ with symmetry in Group 1.
\label{t1}
\end{theorem}
\begin{proof}
\begin{enumerate}
\item Let $(s^*_1, s^*_2, \dots, s^*_n)$ be a Nash equilibrium of a multi-players zero-sum game. This means
\begin{align}
&\min_{s_n\in S_n}\max_{s_i\in S_i}u_i(s_1^*, \dots, s_i, \dots, s_n)\leq \max_{s_i\in S_i}u_i(s_1^*, \dots, s_i, \dots, s_n^*) \label{e3}\\
=&\min_{s_n\in S_n}u_i(s_1^*, \dots, s_i^*, \dots, s_n)\leq \max_{s_i\in S_i}\min_{s_n\in S_n}u_i(s_1^*, \dots, s_i, \dots, s_n).\notag
\end{align}
for Player $i,\ i\in \{1, 2, \dots, n-1\}$.

On the other hand, since
\[\min_{s_n\in S_n}u_i(s_1^*, \dots, s_i, \dots, s_n)\leq u_i(s_1^*, \dots, s_i, \dots, s_n),\]
we have
\[\max_{s_i\in S_i}\min_{s_n\in S_n}u_i(s_1^*, \dots, s_i, \dots, s_n)\leq \max_{s_i\in S_i}u_i(s_1^*, \dots, s_i, \dots, s_n).\]
This inequality holds for any $s_n$. Thus,
\[\max_{s_i\in S_i}\min_{s_n\in S_n}u_i(s_1^*, \dots, s_i, \dots, s_n)\leq \min_{s_n\in S_n}\max_{s_i\in S_i}u_i(s_1^*, \dots, s_i, \dots, s_n).\]
With (\ref{e3}), we obtain
\begin{equation}
\max_{s_i\in S_i}\min_{s_n\in S_n}u_i(s_1^*, \dots, s_i, \dots, s_n)=\min_{s_n\in S_n}\max_{s_i\in S_i}u_i(s_1^*, \dots, s_i, \dots, s_n).\label{t1-1}
\end{equation}
(\ref{e3}) and (\ref{t1-1}) imply
\[\max_{s_i\in S_i}\min_{s_n\in S_n}u_i(s_1^*, \dots, s_i, \dots, s_n)=\max_{s_i\in S_i}u_i(s_1^*, \dots, s_i, \dots, s_n^*),\]
\[\min_{s_n\in S_n}\max_{s_i\in S_i}u_i(s_1^*, \dots, s_i, \dots, s_n)=\min_{s_n\in S_n}u_i(s_1^*, \dots, s^*_i, \dots, s_n).\]
From
\[\min_{s_n\in S_n}u_i(s_1^*, \dots, s_i, \dots, s_n)\leq u_i(s_1^*, \dots, s_i, \dots, s_n^*),\]
and
\[\max_{s_i\in S_i}\min_{s_n\in S_n}u_i(s_1^*, \dots, s_i, \dots, s_n)=\max_{s_i\in S_i}u_i(s_1^*, \dots, s_i, \dots, s_n^*),\]
we have
\[\arg\max_{s_i\in S_i}\min_{s_n\in S_n}u_i(s_1^*, \dots, s_i, \dots, s_n)=\arg\max_{s_i\in S_i}u_i(s_1^*, \dots, s_i, \dots, s_n^*)=s_i^*,\]
for all $i\in \{1, 2, \dots, n-1\}$. $s_i^*$'s  are equal for all $i\in \{1, 2, \dots, n-1\}$.

Also, from
\[\max_{s_i\in S_i}u_i(s_1^*, \dots, s_i, \dots, s_n)\geq u_i(s_1^*, \dots, s_i^*, \dots, s_n),\]
and
\[\min_{s_n\in S_n}\max_{s_i\in S_i}u_i(s_1^*, \dots, s_i, \dots, s_n)=\min_{s_n\in S_n}u_i(s_1^*, \dots, s_i^*, \dots, s_n),\]
we get
\[\arg\min_{s_n\in S_n}\max_{s_i\in S_i}u_i(s_1^*, \dots, s_i, \dots, s_n)=\arg\min_{s_n\in S_n}u_i(s_1^*, \dots, s^*_i, \dots, s_n)=s_n^*,\]
for all $i\in \{1, 2, \dots, n-1\}$.
\end{enumerate}
\end{proof}

Next we show the following theorem. 
\begin{theorem}
Sion's minimax theorem with symmetry in Group 1 implies the existence of a Nash equilibrium which is symmetric in Group 1.
\end{theorem}
\begin{proof}
Let $\tilde{s}$ be a value of $s_j$'s, $j\neq i,\ j\in \{1, 2, \dots, n-1\}$ such that
\[\tilde{s}=\arg\max_{s_i\in S_i}\min_{s_n\in S_n}u_i(\tilde{s}, \dots, s_i, \dots, \tilde{s}, \dots, s_n).\]
Then, we have
\begin{align}
&\max_{s_i\in S_i}\min_{s_n\in S_n}u_i(\tilde{s}, \dots, s_i, \dots, \tilde{s}, \dots, s_n)=\min_{s_n\in S_n}u_i(\tilde{s}, \dots, \tilde{s}, \dots, \tilde{s}, \dots, s_n)\label{t2-1}\\
&=\min_{s_n\in S_n}\max_{s_i\in S_i}u_i(\tilde{s}, \dots, s_i, \dots, \tilde{s}, \dots, s_n).\notag
\end{align}
Since
\[u_i(\tilde{s}, \dots, \tilde{s}, \dots, \tilde{s}, \dots, s_n)\leq \max_{s_i\in S_i}u_i(\tilde{s}, \dots, s_i, \dots, \tilde{s}, \dots, s_n),\]
and
\[\min_{s_n\in S_n}u_i(\tilde{s}, \dots, \tilde{s}, \dots, \tilde{s}, \dots, s_n)=\min_{s_n\in S_n}\max_{s_i\in S_i}u_i(\tilde{s}, \dots, s_i, \dots, \tilde{s}, \dots, s_n),\]
we get
\begin{equation}
\arg\min_{s_n\in S_n}u_i(\tilde{s}, \dots, \tilde{s}, \dots, \tilde{s}, \dots, s_n)=\arg\min_{s_n\in S_n}\max_{s_i\in S_i}u_i(\tilde{s}, \dots, s_i, \dots, \tilde{s}, \dots, s_n).\label{t2-2}
\end{equation}
Since the game is zero-sum,
\[\sum_{i=1}^{n-1}u_i(\tilde{s}, \dots,\tilde{s}, \dots, \tilde{s}, \dots, s_n)=(n-1)u_i(\tilde{s}, \dots,\tilde{s}, \dots, \tilde{s}, \dots, s_n)=-u_n(\tilde{s}, \dots,\tilde{s}, \dots, \tilde{s}, \dots, s_n).\]
Therefore,
\[\arg\min_{s_n\in S_n}u_i(\tilde{s}, \dots,\tilde{s}, \dots, \tilde{s}, \dots, s_n)=\arg\max_{s_n\in S_n}u_n(\tilde{s}, \dots,\tilde{s}, \dots, \tilde{s}, \dots, s_n).\]
Let
\begin{equation}
\hat{s}_n=\arg\min_{s_n\in S_n}u_i(\tilde{s}, \dots,\tilde{s}, \dots, \tilde{s}, \dots, s_n)=\arg\max_{s_n\in S_n}u_n(\tilde{s}, \dots,\tilde{s}, \dots, \tilde{s}, \dots, s_n).\label{t2-3}
\end{equation}
Then, from (\ref{t2-1}) and (\ref{t2-2})
\begin{align*}
&\min_{s_n\in S_n}\max_{s_i\in S_i}u_i(\tilde{s}, \dots,{s}_i, \dots, \tilde{s}, \dots, {s}_n)=\max_{s_i\in S_i}u_i(\tilde{s}, \dots,{s}_i, \dots, \tilde{s}, \dots, \hat{s}_n)\\
&=\min_{s_n\in S_n}u_i(\tilde{s}, \dots,\tilde{s}, \dots, \tilde{s}, \dots, s_n)=u_i(\tilde{s}, \dots,\tilde{s}, \dots, \tilde{s}, \dots, \hat{s}_n).\notag
\end{align*}
Thus,
\begin{equation}
\arg\max_{s_i\in S_i}u_i(\tilde{s}, \dots,{s}_i, \dots, \tilde{s}, \dots, \hat{s}_n)=\tilde{s}\ \mathrm{for\ all}\ i\in \{1, 2, \dots, n-1\}. \label{t2-4}
\end{equation}
(\ref{t2-3}) and (\ref{t2-4}) mean that $(s_1, s_2, \dots, s_{n-1},s_n)=(\tilde{s},\tilde{s},\dots, \tilde{s},\hat{s}_n)$ is a Nash equilibrium in which only Player $n$ may choose a different strategy. 
\end{proof}

\section{Example of relative profit maximizing four-firms oligopoly}

Consider a four-players game. Suppose that the payoff functions of the players are
\begin{align*}
\pi_A=&(a-x_A-x_B-x_C-x_D)x_A-c_Ax_A-\frac{1}{3}[(a-x_A-x_B-x_C-x_D)x_B-c_Bx_B\\
&+(a-x_A-x_B-x_C-x_D)x_C-c_Cx_C+(a-x_A-x_B-x_C-x_D)x_D-c_Dx_D],
\end{align*}
\begin{align*}
\pi_B=&(a-x_A-x_B-x_C-x_D)x_B-c_Bx_B-\frac{1}{3}[(a-x_A-x_B-x_C-x_D)x_A-c_Ax_A\\
&+(a-x_A-x_B-x_C-x_D)x_C-c_Cx_C+(a-x_A-x_B-x_C-x_D)x_D-c_Dx_D],
\end{align*}
\begin{align*}
\pi_C=&(a-x_A-x_B-x_C-x_D)x_C-c_Cx_C-\frac{1}{3}[(a-x_A-x_B-x_C-x_D)x_A-c_Ax_A\\
&+(a-x_A-x_B-x_C-x_D)x_B-c_Bx_B+(a-x_A-x_B-x_C-x_D)x_D-c_Dx_D],
\end{align*}
\begin{align*}
\pi_D=&(a-x_A-x_B-x_C-x_D)x_D-c_Dx_D-\frac{1}{3}[(a-x_A-x_B-x_C-x_D)x_A-c_Ax_A\\
&+(a-x_A-x_B-x_C-x_D)x_B-c_Bx_B+(a-x_A-x_B-x_C-x_D)x_C-c_Cx_C].
\end{align*}
This is a model of relative profit maximization in a four firms Cournot oligopoly with constant marginal costs and zero fixed cost producing a homogeneous good. $x_i,i=A, B, C, D$, are the outputs of the firms. The conditions for maximization of $\pi_i,\ i=A, B, C, D$, are
\[\frac{\partial \pi_A}{\partial x_A}=a-2x_A-(x_B+x_C+x_D)-c_A+\frac{1}{3}(x_B+x_C+x_D)=0,\]
\[\frac{\partial \pi_B}{\partial x_B}=a-2x_B-(x_A+x_C+x_D)-c_B+\frac{1}{3}(x_A+x_C+x_D)=0,\]
\[\frac{\partial \pi_C}{\partial x_C}=a-2x_C-(x_A+x_B+x_D)-c_C+\frac{1}{3}(x_A+x_B+x_D)=0,\]
\[\frac{\partial \pi_D}{\partial x_D}=a-2x_C-(x_A+x_B+x_C)-c_D+\frac{1}{3}(x_A+x_B+x_C)=0.\]
The Nash equilibrium strategies are
\begin{align}
\begin{cases}
&x_A=\frac{2a-5c_A+c_B+c_C+c_D}{8},\label{nash}\\
&x_B=\frac{2a-5c_B+c_A+c_C+c_D}{8},\\
&x_C=\frac{2a-5c_C+c_A+c_B+c_D}{8},\\
&x_D=\frac{2a-5c_D+c_A+c_B+c_C}{8}.
\end{cases}
\end{align}

Next consider maximin and minimax strategies about Player A and Player D. The condition for minimization of $\pi_A$ with respect to $x_D$ is $\frac{\partial \pi_A}{\partial x_D}=0$. Denote $x_D$ which satisfies this condition by $x_D(x_A, x_B,x_C)$, and substitute it into $\pi_A$. Then, the condition for maximization of $\pi_A$ with respect to $x_A$ given $x_D(x_A, x_B,x_C)$, $x_B$ and $x_C$ is
\[\frac{\partial \pi_A}{\partial x_A}+\frac{\partial \pi_A}{\partial x_D}\frac{\partial x_D}{\partial x_A}=0.\]
It is denoted by $\arg\max_{x_A}\min_{x_D}\pi_A$. The condition for maximization of $\pi_A$ with respect to $x_A$ is $\frac{\partial \pi_A}{\partial x_A}=0$. Denote $x_A$ which satisfies this condition by $x_A(x_B, x_C, x_D)$, and substitute it into $\pi_A$. Then, the condition for minimization of $\pi_A$ with respect to $x_D$ given $x_A(x_B, x_C, x_D)$ is
\[\frac{\partial \pi_A}{\partial x_D}+\frac{\partial \pi_A}{\partial x_A}\frac{\partial x_A}{\partial x_D}=0.\]
It is denoted by $\arg\min_{x_D}\max_{x_A}\pi_A$. In our example we obtain
\[\arg\max_{x_A}\min_{x_D}\pi_A=\frac{2a-3c_A+c_D}{8},\ \arg\min_{x_D}\max_{x_A}\pi_A=\frac{6a-3c_A-3c_D-8x_B-8x_C}{8}.\]
Similarly, we get the following results.
\[\arg\max_{x_B}\min_{x_D}\pi_B=\frac{2a-3c_B+c_D}{8},\]
\[\arg\min_{x_D}\max_{x_B}\pi_B=\frac{6a-3c_B-3c_D-8x_A-8x_C}{8},\]
\[\arg\max_{x_C}\min_{x_D}\pi_C=\frac{2a-3c_C+c_D}{8},\]
\[\arg\min_{x_D}\max_{x_C}\pi_C=\frac{6a-3c_C-3c_D-8x_A-8x_B}{8}.\]
If $c_C=c_B=c_A$,
\[\arg\max_{x_A}\min_{x_D}\pi_A=\arg\max_{x_B}\min_{x_D}\pi_B=\arg\max_{x_C}\min_{x_D}\pi_C=\frac{2a-3c_A+c_D}{8}.\]
These are equal to the Nash equilibrium strategies for Firms A, B and C with $c_C=c_B=c_A$ and $c_D\neq c_A$.

When $x_A=x_B=x_C=\frac{2a-3c_A+c_D}{8}$, we have
\[\arg\min_{x_D}\max_{x_A}\pi_A=\arg\min_{x_D}\max_{x_B}\pi_B=\arg\min_{x_D}\max_{x_C}\pi_C=\frac{2a-5c_D+3c_A}{8}.\]
These are equal to the Nash equilibrium strategy for Firm D  with $c_C=c_B=c_A$.

On the other hand, if $c_B=c_A$ and $c_C=c_D$, we have
\[\arg\max_{x_A}\min_{x_D}\pi_A=\arg\max_{x_B}\min_{x_D}\pi_B=\frac{2a-3c_A+c_D}{8}.\]
This is not equal to the Nash equilibrium strategies for Firms A and B, with $c_B=c_A$ and $c_C=c_D\neq c_A$ which are
\[x_A=x_B=\frac{a-2c_A+c_D}{4}\neq \frac{2a-3c_A+c_D}{8}.\]

\section{Concluding Remark}

In this paper we have examined the relation between Sion's minimax theorem for a continuous function and a Nash equilibrium in an asymmetric multi-players zero-sum game in which only one player is different from other players. We have shown that the following two statements are equivalent.
\begin{enumerate}
	\item The existence of a Nash equilibrium,  which is symmetric for players other than one player, implies Sion's minimax theorem for pairs of this player and one of other players with symmetry for the other players.
	\item Sion's minimax theorem for pairs of one player and one of other players with symmetry for the other players implies the existence of a Nash equilibrium which is symmetric for the other players.
\end{enumerate}

As we have shown in Appendix, if there are two aliens, this equivalence does not hold.

\end{document}